\documentclass[12pt,a4paper]{article}
\newcommand{\footremember}[2]{%
    \footnote{#2}
    \newcounter{#1}
    \setcounter{#1}{\value{footnote}}%
}
 
\title{Optimal designs for third-order interactions in paired comparison experiments}
\author{%
  Eric Nyarko \footremember{alley}{Corresponding author. E-mail: eric.nyarko@ovgu.de; nyarkoeric5@gmail.com}
  \date{\textit{\footnotesize Institute for Mathematical Stochastics, University of Magdeburg, PF 4120, D-39016 Magdeburg, Germany
  }}}
\usepackage{amsmath}
\usepackage{booktabs}
\usepackage{dirtytalk}
\usepackage{float,rotating}
\usepackage{diagbox}
\usepackage{slashbox}
\usepackage{appendix}
\usepackage{tikz}
\usepackage{url}
\usepackage{natbib}
\usepackage{array}
\usepackage{longtable,pdflscape,caption}
\usepackage{amsmath,mathtools}
\usepackage{amsthm}
\usepackage{siunitx} 
\usepackage{booktabs}       % for nice rules in table
\usepackage{dcolumn}        % dcolumn to line up decimals
\usepackage[flushleft]{threeparttable}

\usepackage{dcolumn}
\usepackage{lipsum} 

\date{}
\begin{document}
\maketitle
\begin{abstract}\noindent
It is shown how by not losing information on higher order interactions, optimal paired comparison designs involving alternatives of either full or partial profiles to reduce information overload as frequently encountered in applications can be constructed which enable identification of main effects up to third-order interactions when all attributes have general common number of levels.
\end{abstract}
\vspace{0.5mm}

%\newpage
\section{Introduction}
In a paired comparison experiment usually two competing objects (alternatives) $\textbf{i}$ and $\textbf{j}$ say, are presented to a respondent who must trade-off one alternative against the other and state his/her preferences. Data arising from a paired comparison task can either be qualitative \citep[e.g., see][]{bradley1952rank} or quantitative \citep[e.g., see][]{scheffe1952analysis} depending on the response method used. This paper adopts the corresponding method by which the data is generated according to a quantitative response and, when additionally the degree of preference for the alternatives $\textbf{i}$ and $\textbf{j}$ is scored or indicated on a rating scale. For example, scoring alternative $\textbf{i}$ when compared to alternative $\textbf{j}$ on a seven-point scale ($3, 2, 1, 0, -1, -2, -3$) means: $3$; strong preference for $\textbf{i}$ to $\textbf{j}$, $2$; moderate preference for $\textbf{i}$ to $\textbf{j}$, $1$; slight preference for $\textbf{i}$ to $\textbf{j}$, $0$; no preference, $-1$; slight preference for $\textbf{j}$ to $\textbf{i}$, $-2$; moderate preference for $\textbf{j}$ to $\textbf{i}$ and $-3$; strong preference for $\textbf{j}$ to $\textbf{i}$. \par

The method of paired comparisons can be employed in a hypothetical (occasionally real) situation. It has received considerable attention in many fields of applications like psychology, health economics, environmental valuation, transportation economics and marketing to study people's preferences for goods or services. However, in application a situation may arise where information on main effects and interactions is worthwhile. This paper is motivated by the situation where the designs enable identification of main effects and up to third-order interactions, which may be not of primary application interest but theoretically worthwhile \citep[e.g., see][]{quenouille1971paired,el1984optimal,lewis1985paired,elrod1992empirical,collins2009design}.
 \par  
In order to reduce information overload as frequently encountered in applications when a respondent has to compare alternatives described by a large number of attributes, comparisons are often restricted to only a subset of the attributes with potentially different levels and the remaining attributes are usually set to zero. The number of attributes that are shown in this restricted settings is called the profile strength, and the set of alternatives described by this profile strength is known as partial profiles \citep{chrzan2010using,kessels2011bayesian}.  \par
The aim of this paper is to construct designs that estimate main effects and up to third-order interactions for the situation of both full and partial profiles when all attributes have the same number of levels. Results for main effects plus first- plus second-order interactions are provided \citep{grasshoff2004optimal,grasshoff2003optimal,schwabe2003optimal,nyarko2019optimal,eric2019optimal4444}.\par
We mention that the designs considered in this paper possess large number of comparisons but can serve as a benchmark to judge the efficiency of competing designs as well as a starting point to construct designs which share the property of optimality and can be realized with a reasonable number of comparisons.\par

This paper is organized as follows. A general model is introduced in Section $2$ for linear paired comparisons. Section $3$ provides the third-order interactions model for full and partial profiles. The optimal design results obtained are discussed in Section $4$. All proofs are deferred to the Appendix.

\section{General setting} 
Suppose in an experimental situation there are $K$ attributes such that the $k$th attribute has $i_k$ levels $(i_k=1,\dots,v)$. Then in this setting any observation (utility) $\tilde{Y}_{na}(\textbf{i})$ of a single alternative $\textbf{i}= (i_{1},\dots,i_{K})$ where $i_k$ is the component of the $k$th attribute, $k=1,\dots,K$ within a pair of alternatives ($a=1,2$) can be formalized by a general linear model\par

\begin{equation}\label{eq:1}
\begin{split}
\tilde{Y}_{na}(\textbf{i})&= \mu_{n}+\textbf{f}(\textbf{i})^{\top} \boldsymbol{\beta} + \tilde{\varepsilon}_{na},
\end{split}
\end{equation}
where $\mu_n$ is the block effect, the index $n$ denotes the $n$th presentation, $n=1,\dots,N$, and the alternative $\textbf{i}$ is selected from a set $\mathcal{I}$ of possible realizations of the alternatives. Here $\textbf{f}$ is the regression function, the mean response $\mu_n+\textbf{f}(\textbf{i})^{\top}\boldsymbol{\beta}$ is assumed to be known, $\boldsymbol{\beta}$ is the vector of unknown parameters of interest, and the random error $\tilde{\varepsilon}_{na}$ is assumed to be uncorrelated with constant variance and zero mean.  \par

Without loss of generality, we consider paired comparison experiments where the utilities for the alternatives are not directly observed. Only observations $Y_n(\textbf{i},\textbf{j})=\tilde{Y}_{n1}(\textbf{i})-\tilde{Y}_{n2}(\textbf{j})$ of the amount of preference are available for comparing pairs $(\textbf{i},\textbf{j})$ of alternatives $\textbf{i}$ and $\textbf{j}$ which are chosen from the design region $\mathcal{X}=\mathcal{I}\times \mathcal{I}$. In this case the utilities for the alternatives are properly described by the linear paired comparison model 
\begin{equation} \label{eq:2}
\begin{split}
Y_n(\textbf{i, j})=(\textbf{f}(\textbf{i})-\textbf{f}(\textbf{j}))^{\top}\boldsymbol\beta+\varepsilon_{n}, \\
\end{split}
\end{equation}
where $\textbf{f}(\textbf{i})-\textbf{f}(\textbf{j})$ is the derived regression function and the random errors $\varepsilon_{n}(\textbf{i}, \textbf{j})=\tilde{\varepsilon}_{n1}(\textbf{i})-\tilde{\varepsilon}_{n2}(\textbf{j})$ associated with the different pairs $(\textbf{i}, \textbf{j})$ are assumed to be uncorrelated with constant variance and zero mean. \par

The performance of the statistical analysis depends on the  pairs $(\textbf{i}, \textbf{j})$ that are chosen from the design region $\mathcal{X}=\mathcal{I}\times \mathcal{I}$. The choice of such pairs $(\textbf{i}_1,\textbf{j}_1),\dots,(\textbf{i}_N,\textbf{j}_N)\in\mathcal{X}$ makes up the design for the study. 
The goodness of a design is measured by its information matrix
\begin{equation}\label{eq:3}
\textbf{M}((\textbf{i}_1, \textbf{j}_1),\dots,(\textbf{i}_N, \textbf{j}_N))=\sum_{n=1}^{N}\textbf{M}((\textbf{i}_n, \textbf{j}_n)),
\end{equation}
where $\textbf{M}((\textbf{i},\textbf{j}))=(\textbf{f}(\textbf{i})-\textbf{f}(\textbf{j}))(\textbf{f}(\textbf{i})-\textbf{f}(\textbf{j}))^{\top}$ is the elemental information of a single pair $(\textbf{i},\textbf{j})$.\par

In the optimal design literature two types of designs are studied: approximate or continuous designs and exact designs. This paper focuses on the former designs \citep[e.g., see][]{kiefer1959optimum} which are essentially probability measures defined on a design region. Every approximate design $\xi$ which assigns only rational weights $\xi(\textbf{i},\textbf{j})$ to all pairs $(\textbf{i},\textbf{j})\in\mathcal{X}$ can be realized as an exact design $\xi_N$ of size $N$ consisting of the pairs $(\textbf{i}_1,\textbf{j}_1),\dots,(\textbf{i}_N,\textbf{j}_N)$. \par

The information matrix of an approximate design $\xi$ in the linear paired comparison model is given by
\begin{equation}\label{eq:4}
\textbf{M}(\xi)=\sum_{n=1}^{N}\xi(\textbf{i}_n,\textbf{j}_n)(\textbf{f}(\textbf{i}_n)-\textbf{f}(\textbf{j}_n))(\textbf{f}(\textbf{i}_n)-\textbf{f}(\textbf{j}_n))^{\top},
\end{equation} 
which is proportional to the inverse of the covariance matrix for the best linear unbiased estimator of the parameter vector $\boldsymbol{\beta}$. The normalized information matrix $\textbf{M}(\xi_N)$ for an exact design $\xi_N$ coincides with the information matrix $\textbf{M}(\xi)$ of the corresponding approximate design $\xi$.\par
The $D$-optimality criterion is considered, which can be regarded as a scalar measure of design quality. An approximate design $\xi^{\ast}$ is $D$-optimal if it maximizes the determinant of the information matrix, that is, if $\mathrm{det} \textbf{M}(\xi^{\ast})$ $\geq$ $\mathrm{det} \textbf{M}(\xi)$ for every approximate design $\xi$ on $\mathcal{X}$.  \par

It is worthwhile mentioning that under the indifference assumption of equal choice probability the designs considered in this paper carry over to the \citet{bradley1952rank} type choice experiments \citep[e.g., see][]{grossmann2015handbook}.

\section{Third-order interactions model}
In this section, we give a characterization of the third-order interactions model under consideration. Corresponding results for the particular case of binary attributes can be found  \citep{eric2019optimal4444}. The instances when information on main effects, first- and second-order interactions is worthwhile can be found \citep{van1987optimal, grasshoff2003optimal,nyarko2019optimal}. \par

Analogously to \citet{nyarko2019optimal}, we first proceed with the situation of full profiles when all attributes, say, $K$ for which the $k$th attribute has $i_k$ levels $(i_k=1,\dots,v)$ are presented. For the present setting of paired comparisons, we denote the first alternative by $\textbf{i}=(i_1,\dots,i_K)$ and the second alternative by $\textbf{j}=(j_1,\dots,j_K)$. For each attribute $k$ the corresponding regression functions $\mathbf{f}_k=\mathbf{f}$ coincide with the one-way layout \citep[see][Section~3]{grasshoff2004optimal}.\par

In the presence of up to third-order interactions direct responses $\tilde{Y}_{na}$ at alternative $\textbf{i}=(i_1,\ldots,i_K)$ can be modeled as
\begin{align}\label{eq:full_direct}
\tilde{Y}_{na}(\textbf{i})&=\mu_n+\sum_{k=1}^{K}\textbf{f}(i_k)^\top \boldsymbol{\beta}_k+\sum_{k<\ell}(\textbf{f}(i_k)\otimes \textbf{f}(i_\ell))^\top\boldsymbol{\beta}_{k\ell} \nonumber\\
&\qquad+\sum_{k<\ell<m}(\textbf{f}(i_k)\otimes \textbf{f}(i_\ell)\otimes \textbf{f}(i_m))^\top\boldsymbol{\beta}_{k\ell m}
\nonumber\\
&\qquad+\sum_{k<\ell<m<r}(\textbf{f}(i_k)\otimes \textbf{f}(i_\ell)\otimes \textbf{f}(i_m)\otimes \textbf{f}(i_r))^\top\boldsymbol{\beta}_{k\ell mr}+\tilde{\varepsilon}_{na},
\end{align}
for full profiles, where $\otimes$ denotes the Kronecker product of vectors or matrices, $\boldsymbol{\beta}_k=(\beta_{i_k}^{(k)})_{i_k=1,\dots,v-1}$ denotes the main effect of the $k$th attribute, $\boldsymbol{\beta}_{k\ell}=(\beta_{i_ki_{\ell}}^{(k\ell)})_{i_k=1,\dots,v-1,~i_{\ell}=1,\dots,v-1}$ is the first-order interaction of the $k$th and $\ell$th attribute, $\boldsymbol{\beta}_{k\ell m}=(\beta_{i_ki_{\ell}i_m}^{(k\ell m)})_{i_k=1,\dots,v-1,~i_{\ell}=1,\dots,v-1,~i_{m}=1,\dots,v-1}$ is the second-order interaction of the $k$th, $\ell$th and $m$th attribute, and $\boldsymbol{\beta}_{k\ell mr}=(\beta_{i_ki_{\ell}i_mi_r}^{(k\ell mr)})_{i_k=1,\dots,v-1,}$ $_{i_{\ell}=1,\dots,v-1,~i_{m}=1,\dots,v-1,~i_{r}=1,\dots,v-1}$ is the third-order interaction of the $k$th, $\ell$th, $m$th and $r$th attribute.
The vectors $(\boldsymbol{\beta}_k)_{1\leq k\leq K}$ of main effects have parameter $p_1=K(v-1)$, $(\boldsymbol{\beta}_{k\ell})_{1\leq k<\ell\leq K}$ of first-order interactions have parameter $p_2=(1/2)K(K-1)(v-1)^{2}$, $(\boldsymbol{\beta}_{k\ell m})_{1\leq k<\ell<m\leq K}$ of second-order interactions have parameter $p_3=(1/6)K(K-1)(K-2)(v-1)^{3}$, and $(\boldsymbol{\beta}_{k\ell mr})_{1\leq k<\ell<m<r\leq K}$ of third-order interactions have parameter $p_4=(1/24)K(K-1)(K-2)(K-3)(v-1)^{4}$.
Hence the complete parameter vector 
\begin{align}\label{eqtg:4.}
\boldsymbol{\beta}=((\boldsymbol{\beta}_k)_{k=1,\dots,K}^\top,(\boldsymbol{\beta}_{k\ell})_{k<\ell}^\top,(\boldsymbol{\beta}_{k\ell m})^\top_{k<\ell<m},(\boldsymbol{\beta}_{k\ell mr})^\top_{k<\ell<m<r})^\top,
\end{align}
 has dimension $p=p_1+p_2+p_3+p_4$. 
The corresponding $p$-dimensional vector $\textbf{f}$ of regression functions is given by
\begin{align}\label{eq:10}
\textbf{f}(\textbf{i})&=(\textbf{f}(i_1)^\top,\dots,\textbf{f}(i_K)^\top,\textbf{f}(i_1)^\top\otimes\textbf{f}(i_2)^\top,\dots,\textbf{f}(i_{K-1})\otimes\textbf{f}(i_K)^\top,  \nonumber\\
&\ ~~ \textbf{f}(i_1)^\top\otimes\textbf{f}(i_2)^\top\otimes\textbf{f}(i_3)^\top,\dots,\textbf{f}(i_{K-2})^\top\otimes\textbf{f}(i_{K-1})^\top\otimes\textbf{f} (i_K)^\top,\nonumber\\
&\ ~~\textbf{f}(i_1)\otimes\textbf{f}(i_2)\otimes\textbf{f}(i_3)\otimes\textbf{f}(i_4),\dots,\textbf{f}(i_{K-3})\otimes\textbf{f}(i_{K-2})\otimes\textbf{f}(i_{K-1})\otimes\textbf{f}(i_K)^\top),^\top 
\end{align}
where the first $K$ components $\textbf{f}(i_{1}),\dots,\textbf{f}(i_{K})$ are associated with the main effects, the second components $\textbf{f}(i_1)\otimes\textbf{f}(i_2),\dots,\textbf{f}(i_{K-1})$ $\otimes\textbf{f}(i_K)$ are associated with the first-order interactions, the third components $\textbf{f}(i_1)\otimes\textbf{f}(i_2)\otimes\textbf{f}(i_3),\dots,\textbf{f}(i_{K-2})\otimes\textbf{f}(i_{K-1})\otimes\textbf{f}(i_K)$ are associated with the second-order interactions and the remaining components $\textbf{f}(i_1)\otimes\textbf{f}(i_2)\otimes\textbf{f}(i_3)\otimes\textbf{f}(i_4),\dots,\textbf{f}(i_{K-3})\otimes\textbf{f}(i_{K-2})\otimes\textbf{f}(i_{K-1})\otimes\textbf{f}(i_K)$ of $\textbf{f}(\textbf{i})$ are associated with the third-order interactions.\par
As was already pointed out, in order to reduce information overload as frequently encounted in applications when one has to compare alternatives described by a large number of attributes, comparisons are often restricted to only a subset of the attributes with potentially different levels and the remaining attributes are usually set to zero. The number of attributes in this restricted settings is referred to as the profile strength $S$ and alternatives described by $S$ are referred to as partial profiles. In this setting a direct observation can be described by \eqref{eq:full_direct} even for a partial profile $\textbf{i}$ from the set 
\begin{equation}\label{eq:12}
\begin{split}
\mathcal{I}^{(S)}=&\{\textbf{i};\ i_{k}\in\{1,\dots,v\} \textrm{ for at most $S$ indices $k$} \}.
\end{split}
\end{equation}
Here $\mathcal{I}^{(K)}=\mathcal{I}^{(S)}$ in the case of full profiles ($S=K$).\par

Moreover, for observations in linear paired comparisons the resulting model is given by
\begin{align}\label{eq:11}
Y_{n}(\textbf{i},\textbf{j})&=\sum_{k=1}^{K}(\textbf{f}(i_k)-\textbf{f}(j_k))^\top \boldsymbol{\beta}_k+\sum_{k<\ell}((\textbf{f}(i_k)\otimes \textbf{f}(i_\ell))-(\textbf{f}(j_k)\otimes \textbf{f}(j_\ell)))^\top\boldsymbol{\beta}_{k\ell} \nonumber\\
&\qquad+\sum_{k<\ell<m}((\textbf{f}(i_k)\otimes \textbf{f}(i_\ell)\otimes \textbf{f}(i_m))-(\textbf{f}(j_k)\otimes \textbf{f}(j_\ell)\otimes \textbf{f}(j_m)))^\top\boldsymbol{\beta}_{k\ell m}  \nonumber\\
&\qquad+\sum_{k<\ell<m<r}((\textbf{f}(i_k)\otimes \textbf{f}(i_\ell)\otimes \textbf{f}(i_m)\otimes \textbf{f}(i_r))  \nonumber\\
&\qquad\qquad\qquad\qquad-(\textbf{f}(j_k)\otimes \textbf{f}(j_\ell)\otimes \textbf{f}(j_m)\otimes \textbf{f}(j_r)))^\top\boldsymbol{\beta}_{k\ell m r}+\varepsilon_{n}.
\end{align}
Accordingly, the design region can be specified as

\begin{equation}\label{eq:12}
\begin{split}
\mathcal{X}^{(S)}&=\{(\textbf{i},\textbf{j});~i_{k}, j_{k}\in\{1,\dots,v\} \ \textrm{for at most $S$ indices $k$}\}.
\end{split}
\end{equation}
Also, here the design region $\mathcal{X}^{(K)}=\mathcal{I}^{(K)}\times\mathcal{I}^{(K)}$ in the case of full profiles ($S=K$), where all pairs of alternatives are described by all attributes.

\section{Optimal designs}
In the present setting we construct designs that estimate main effects and higher-order interactions, which is motivated by the work of \citet{quenouille1971paired}, \citet{lewis1985paired} and \citet{el1984optimal}, among others. In particular, optimal designs will be investigated under the third-order interactions paired comparison model \eqref{eq:11} with corresponding regression functions $\textbf{f}(\textbf{i})$ given by \eqref{eq:10}. Here we point out that a minimum of four attributes are required to enable identifiability of the interactions. \par
Following \citet{grasshoff2003optimal}, we define $d$ as the comparison depth which describes the number of attributes in which the two alternatives presented differ satisfying $d=0,1,\dots, S$. For this situation the design region $\mathcal{X}^{(S)}$ in \eqref{eq:12} can be partitioned into disjoint sets 
\begin{equation}\label{eq:13}
\mathcal{X}^{(S)}_{d}=\{(\textbf{i},\textbf{j})\in\mathcal{X}^{(S)};\ i_{k}\neq j_{k} \textrm{ for exactly $d$ components}\}.
\end{equation}
These sets constitute the orbits with respect to permutations of both the levels $i_k,j_k=1,\dots,v$ within the attributes as well as among attributes $k=1,\dots,K$, themselves. Here the problem of finding $D$-optimal designs is restricted to the class of invariant designs \citep[see][Section 3.2]{1996optimum} which are uniform on the orbits of fixed comparison depth $d\leq S$. For the corresponding set we let $N_{d}={K \choose S}{S \choose d}v^S(v-1)^d$ be the number of different pairs which vary in exactly $d$ attributes and denote $\bar{\xi}_{d}$ as the uniform approximate design which assigns equal weights $\bar{\xi}_{d}(\textbf{i},\textbf{j})=1/N_{d}$ to each pair $(\textbf{i},\textbf{j})$ in $\mathcal{X}^{(S)}_{d}$ and weight zero to all remaining pairs in $\mathcal{X}^{(S)}$. In the following we present the information matrix for the corresponding invariant designs. 
\newtheorem{lemma}{Lemma}
\begin{lemma}\label{lemma1}
The uniform design $\bar\xi_{d}$ on the set $\mathcal{X}^{(S)}_{d}$ of comparison depth $d$ has block diagonal information matrix
\begin{equation*}
\mathbf{M}(\bar\xi_{d})=diag(h_q(d)\mathbf{Id}_{p_q}\otimes \mathbf{M}^{\otimes q})_{q=1,\dots,4},
\end{equation*}
where $\mathbf{M}^{\otimes q}$ denotes the $q$-fold Kronecker product of  $\mathbf{M}$ and
\small{
\begin{equation*}
\begin{split}
&h_{1}(d)=\frac{d}{K},~h_{2}(d) =\frac{d((d-1)(v-2)+2(S-d)(v-1))}{2vK(K-1)},\\
&h_{3}(d) =\frac{d\lambda_1(d)}{4v^{2}K(K-1)(K-2)}, ~h_{4}(d) =\frac{d\lambda_2(d)}{8v^{3}K(K-1)(K-2)(K-3)}, \\
&\lambda_1(d)=(d-1)(d-2)(v^2-3v+3)+3(S-d)(d-1)(v-1)(v-2)\\
&~~~~~~~~~~~+3(S-d)(S-d-1)(v-1)^2,  \\
&\lambda_2(d)=(d-1)(d-2)(d-3)(v^3-4v^2+6v-4)\\
&~~~~~~~~~~~+4(S-d)(d-1)(d-2)(v^2-3v+3)(v-1)  \\
&~~~~~~~~~~~+ 6(S-d)(S-d-1)(d-1)(v-1)^2(v-2)\\
&~~~~~~~~~~~+4(S-d)(S-d-1)(S-d-2)(v-1)^3.  
\end{split}
\end{equation*}
}
\end{lemma}
Here, $\mathbf{Id}_m$ is the identity matrix of order $m$ for every $m$ and $\mathbf{M}=\frac{2}{v-1}(\mathbf{Id}_{v-1}+\mathbf{1}_{v-1}\mathbf{1}^{\top}_{v-1})$ is the information matrix of the one-way layout. 
\par

Further an invariant design $\bar{\xi}$ can be written as a convex combination $\bar{\xi}=\sum^{S}_{d=1}w_{d}\bar{\xi}_{d}$ of uniform designs on the comparison depths $d$ with weights $w_{d}\geq 0$, $\sum^{S}_{d=1}w_{d}=1$. Hence, also every invariant design has diagonal information matrix:
\begin{lemma}\label{lemma2}
The invariant design $\bar{\xi}$ on the design region $\mathcal{X}^{(S)}$ has information matrix of the form
\begin{equation*}
\mathbf{M}(\bar\xi)=diag(h_q(\bar\xi)\mathbf{Id}_{p_q}\otimes \mathbf{M}^{\otimes q})_{q=1,\dots,4},
\end{equation*}
where $h_{q}(\bar{\xi})=\sum_{d=1}^Sw_dh_q(d)$, $q=1,2,3,4$.  
\end{lemma}
First we consider optimal designs for the main effects, the first-, second- and third-order interaction terms separately by maximizing the entries $h_{q}(d)$ in the corresponding information matrix $\mathbf{M}(\bar{\xi}_{d})$. The resulting designs can optimize every invariant design criterion if interest is in the full parameter vector of the corresponding main effects and interactions. Specifically, in the following Table \ref{Tab4.4h3PPROFILE} the values of $d^{\ast}$ recorded in brackets where the first, second and third entries correspond to the first-, second- and third-order interactions respectively, were obtained by first calculating the values of $h_q(d)$ and determining the maximum. Zero entries in the table indicate that the minimum number of attributes required for identifiability of the interactions is not available. It is worthwhile mentioning that the optimal comparison depth $d^{\ast}=S$ for the case of main effects, while for the case of first-order interactions $d^*=S/2$ for $S$ even as well as $d^*=(S-1)/2$ for $S$ odd in the presence of very moderate values of $v$ ($v=2$, for example) and $d^*=S-1$ for sufficiently large values of $v$ ($v=20$, for example). Further, for the case of second-order interactions $d^{\ast}=S$ but this is not true for the situation where $S=K=3$.  Moreover, for the case of third-order interactions $d^{\ast}=S-3$ for sufficiently large values of $v$ ($v=20$, for example). This means that for the corresponding main effects and interactions only those pairs of alternatives should be used which differ in the comparison depth $d^*$ subject to the profile strength $S$. %\vspace{-5ex}
\begin{landscape}
\begin{table}[H]
\centering
\caption{Values of the optimal comparison depths of the $D$-optimal uniform designs for $S=K$.}\label{Tab4.4h3PPROFILE}  
\begin{tabular}{ccccccccccc}\toprule
    \multicolumn{10}{c}{$v$}&\\
    \cline{2-11}
  $S$ & 2       &    3     &  4 & 5  &  6 & 7    &    8& 9&10& 20                         \\
    \hline
2&(1, 0, 0)&(1, 0, 0)&(1, 0, 0) &(1, 0, 0)&(1, 0, 0)&(1, 0, 0)&(1, 0, 0)&(1, 0, 0)&(1, 0, 0)&(1, 0, 0)\\   
3&(1, 1, 0)&(2, 1, 0)&(2, 1, 0) &(2, 1, 0)&(2, 1, 0)&(2, 1, 0)&(2, 1, 0)&(2, 1, 0)&(2, 1, 0)&(2, 1, 0)\\ 
4&(2, 4, 1)&(2, 4, 1)&(3, 4, 1) &(3, 4, 1)&(3, 4, 1)&(3, 4, 1)&(3, 4, 1)&(3, 4, 1)&(3, 4, 1)&(3, 4, 1)\\ 
5&(2, 5, 1)&(3, 5, 1)&(3, 5, 1) &(4, 5, 2)&(4, 5, 2)&(4, 5, 2)&(4, 5, 2)&(4, 5, 2)&(4, 5, 2)&(4, 5, 2)\\ 
6&(3, 6, 1)&(4, 6, 2)&(4, 6, 2) &(4, 6, 2)&(5, 6, 2)&(5, 6, 3)&(5, 6, 3)&(5, 6, 3)&(5, 6, 3)&(5, 6, 3)\\ 
7&(3, 7, 1)&(4, 7, 2)&(5, 7, 3) &(5, 7, 3)&(5, 7, 3)&(6, 7, 3)&(6, 7, 3)&(6, 7, 4)&(6, 7, 4)&(6, 7, 4)\\ 
8&(4, 8, 2)&(5, 8, 3)&(6, 8, 3) &(6, 8, 4)&(6, 8, 4)&(6, 8, 4)&(7, 8, 4)&(7, 8, 4)&(7, 8, 4)&(7, 8, 5)\\ 
9&(4, 9, 2)&(6, 9, 3)&(6, 9, 4) &(7, 9, 4)&(7, 9, 5)&(7, 9, 5)&(7, 9, 5)&(8, 9, 5)&(8, 9, 5)&(8, 9, 6)\\ 
10&(5, 10, 2)&(6, 10, 4)&(7, 10, 4) &(8, 10, 5)&(8, 10, 5)&(8, 10, 6)&(8, 10, 6)&(8, 10, 6)&(9, 10, 6)&(9, 10, 7)\\   \bottomrule
\end{tabular}
    \end{table}
\end{landscape}
\newtheorem{remark}{Remark}
\begin{remark}\label{theorem1}
The uniform design $\bar{\xi}_{S}$ on the largest possible comparison depth $S$ is $D$-optimal for the vector of main effects $(\boldsymbol{\beta}_{1}$ $\dots,$ $\boldsymbol{\beta}_{K})^{\top}$.
\end{remark}

\begin{remark}\label{theorem2}
The uniform design $\bar{\xi}_{d^*}$ is $D$-optimal for the vector of first-order interaction effects $(\boldsymbol{\beta}_{k\ell})_{k<\ell}^{\top}$.
\end{remark}

\begin{remark}\label{theorem3}
The uniform design $\bar{\xi}_{d^*}$ is $D$-optimal for the vector of second-order interaction effects $(\boldsymbol{\beta}_{k\ell})_{k<\ell}^{\top}$.
\end{remark}

\newtheorem{theorem}{Theorem}

\begin{remark}\label{thrm22}
The uniform design $\bar{\xi}_{d^*}$ is $D$-optimal for the vector of third-order interaction effects $(\boldsymbol{\beta}_{k\ell mr})_{k<\ell<m<r}^{\top}$.
\end{remark}

For the corresponding situation of main effects up to third-order interactions, obviously no design exists which simultaneously optimizes the information of the whole parameter vector $\boldsymbol{\beta}=((\boldsymbol{\beta}_k)_{k=1,\dots,K}^\top,(\boldsymbol{\beta}_{k\ell})_{k<\ell}^\top,(\boldsymbol{\beta}_{k\ell m})^\top_{k<\ell<m},$  $(\boldsymbol{\beta}_{k\ell mr})^\top_{k<\ell<m<r})^\top$. As a result, we restrict attention to the $D$-criterion to derive optimal designs for the corresponding whole parameter vector. To reach this goal, it suffices to mention that for invariant designs $\bar{\xi}$ the variance function $V((\textbf{i},\textbf{j}),\bar{\xi})=(\textbf{f}(\textbf{i})-\textbf{f}(\textbf{j}))^{\top}\textbf{M}(\bar{\xi})^{-1}(\textbf{f}(\textbf{i})-\textbf{f}(\textbf{j}))$ with nonsingular information matrix $\textbf{M}(\bar\xi)$ is also invariant with respect to permutations and, hence, constant on the orbits $\mathcal{X}^{(S)}_{d}$ of fixed comparison depth $d$. According to \citet{kiefer1960equivalence}, a design $\xi^{\ast}$ is $D$-optimal if $V((\textbf{i},\textbf{j}),\xi^{\ast})\leq p$.
Hence, the value of the variance function for an invariant design $\bar{\xi}$ evaluated at comparison depth $d$ may be denoted by $V(d,\bar{\xi})$, say, where $V(d,\bar{\xi})=V((\textbf{i},\textbf{j}),\bar{\xi})$ on $\mathcal{X}^{(S)}_{d}$. 
%%%
\begin{theorem}\label{thrm4}
For every invariant design $\bar{\xi}$ the variance function $V(d,\bar{\xi})$ is given by  
\begin{equation*}
\begin{split}
V(d, \bar{\xi})&=d(v-1)\Bigg(\frac{1}{h_{1}(\bar{\xi})}+\frac{v-1}{4vh_{2}(\bar{\xi})}\begin{pmatrix}(d-1)(v-2)+2(S-d)(v-1)\end{pmatrix}\\
&\qquad\qquad\qquad+\frac{(v-1)^2}{24v^2h_{3}(\bar{\xi})}\big((d-1)(d-2)(v^2-3v+3)  \\
&\qquad\qquad\qquad\qquad+3(S-d)(d-1)(v-1)(v-2) \\
&\qquad\qquad\qquad\qquad+3(S-d)(S-d-1)(v-1)^2\big)\\
&\qquad\qquad\qquad+\frac{(v-1)^4}{192v^3h_{4}(\bar{\xi})}\Big((d-1)(d-2)(d-3)(v^3-4v^2+6v-4) \\
&\qquad\qquad\qquad\qquad+4(S-d)(d-1)(d-2)(v^2-3v+3)(v-1)  \\
&\qquad\qquad\qquad\qquad+6(S-d)(S-d-1)(d-1)(v-1)^2(v-2)  \\
&\qquad\qquad\qquad\qquad+4(S-d)(S-d-1)(S-d-2)(v-1)^3\Big)\Bigg).
\end{split}
\end{equation*}
\end{theorem}
On a single comparison depth $d^{\prime}$ the corresponding variance function $V(d, \bar{\xi})$ simplifies:
\newtheorem{corollary}{Corollary}
\begin{corollary}\label{cor_thrm4}
For a uniform design $\bar\xi_{d^{\prime}}$ the variance function is given by 
\begin{equation*}
\begin{split}
V(d,\bar{\xi}_{d^{\prime}})=\frac{d}{d^{\prime}}\begin{pmatrix}p_{1}+p_{2}\frac{(d-1)(v-2)+2(S-d)(v-1)}{(d^{\prime}-1)(v-2)+2(S-d^{\prime})(v-1)}+p_{3}\frac{\lambda_1(d)}{\lambda_1(d^{\prime})}+p_{4}\frac{\lambda_2(d)}{\lambda_2(d^{\prime})}\end{pmatrix}.
\end{split}
\end{equation*}
\end{corollary}
%\begin{proof}
%In view of Theorem~\ref{thrm19} it is sufficient to note that the representation of the variance function follows immediately by inserting the values of $h_{q}(\xi_d)$ from Lemma \ref{lem9} and $p_q={K \choose q}(v-1)^{q}$, $q=1,2,3$. 
%\end{proof}
For $d=d^{\prime}$ we obtain $V(d,\bar{\xi}_d)=p_1+p_2+p_3+p_4=p$ which shows the $D$-optimality of $\bar{\xi}_d$ on $\mathcal{X}^{(S)}_d$ in view of the Kiefer-Wolfowitz equivalence theorem \citep{kiefer1960equivalence}.  \par

The following theorem gives the maximum number of comparison depths required for a $D$-optimal design by virtue of the equivalence theorem.
\begin{theorem}\label{theorem5}
The $D$-optimal design $\xi^{\ast}$ is supported on, at most, four different comparison depths. An example of such depths is $d^{\ast},d^*_1,d^*+1$ and $d^{\ast}_1+1$.
\end{theorem}
The optimal design for the full parameter vector are presented in Table \ref{tab4.7}, where numerical computations indicate that at most two different comparison depths $d^*$ and $d^{\ast}_1$ may be required for $D$-optimality. The corresponding optimal designs with their optimal comparison depths $d^{\ast}$ (in boldface) and their corresponding weights $w^{\ast}_{d^{\ast}}$ for various choices of attributes $K=5,\dots,10$ and levels $v=2,\dots,8$ are exhibited in Table \ref{tab4.7} where entries of the form $(d^{\ast},d^*_1, w_{d^*}^{\ast})$ indicate that invariant designs $\xi^*=w_{d^*}^*\xi_{d^*}+(1-w_{d^*}^*)\xi_{d^*_1}$ have to be considered, while for single entries $d^{\ast}$ the optimal design $\xi^*= \xi^*_{d^*}$ has to be considered which is uniform on the optimal comparison depth $d^{\ast}$. For the particular case $S=K=4$ of full profiles the uniform design on all possible comparisons proves to be optimal \citep[see][Theorem~4]{grasshoff2003optimal}. It is worth noting that for the case of binary attributes, the corresponding results can be found in \citet{eric2019optimal4444}. The values of the normalized variance function $V(d,\xi^{\ast})/p$ which show $D$-optimality of the design $\xi^{\ast}$ in view of the \citet{kiefer1960equivalence} equivalence theorem are recorded in Table \ref{tab4} in the Appendix, where maximal values less than or equal to $1$ establish optimality. It can be seen that for moderate values of $v$ ($v=2$, for example) two types of pairs have to be used in which the numbers of distinct attributes are symmetric with respect to about half of the profile strength to obtain a $D$-optimal design for the whole parameter vector, while for large values of $v$ only one type of pair is sufficient but this is not true for the case $S=6$ and $v=3$.
% \begin{landscape}
 
\begin{table}[H]
\centering
\caption{Optimal comparison depths and optimal weights for $S=K$.}\label{tab4.7} 
\resizebox{!}{.11\paperheight}{
\begin{tabular}{cccccccc}\toprule
    \multicolumn{6}{c}{$v$}&\\
    \cline{2-8}
    $S$       &    2     &  3 & 4  &  5  & 6    &    7&  8                                \\
    \hline
    
5&(\textbf{2}, \textbf{4},\ 0.665)&\textbf{2}&\textbf{2}&\textbf{2}&\textbf{2}&\textbf{2}&\textbf{2}\\

6&(\textbf{2}, \textbf{5},\ 0.714)&(\textbf{2}, \textbf{5},\ 0.878)&\textbf{3}&\textbf{3}&\textbf{3}&\textbf{3}&\textbf{3}\\

7&(\textbf{2}, \textbf{6},\ 0.750)&\textbf{3}&\textbf{3}&\textbf{3}&\textbf{3}&\textbf{4}&\textbf{4}\\

8&(\textbf{3}, \textbf{6},\ 0.667)&\textbf{3}&\textbf{4}&\textbf{4}&\textbf{4}&\textbf{4}&\textbf{4}\\

9&(\textbf{3}, \textbf{7},\ 0.700)&\textbf{4}&\textbf{4}&\textbf{5}&\textbf{5}&\textbf{5}&\textbf{5}\\

10&(\textbf{3}, \textbf{8},\ 0.727)&\textbf{4}&\textbf{5}&\textbf{5}&\textbf{6}&\textbf{6}&\textbf{6}\\ \bottomrule
\end{tabular}}
\end{table}
% \end{landscape} 
 
\vspace{5mm}
\noindent
\textbf{Acknowledgements}
This work was partially supported by Grant - Doctoral Programmes in Germany, $2016/2017~(57214224)$ - of the German Academic Exchange Service (DAAD).
%%%%%%%%%%%%%%%%%%%
\bibliographystyle{apa}     
\bibliography{reference4th}   
%%%%%%%%%%%%%%%%%%%
\begin{appendices}
\section*{APPENDIX}
\begin{proof}[Proof of Lemma~\ref{lemma1}]
First we note that the quantities $h_q(d)$ for $q=1,2,3$ can be obtained as in \citep{grasshoff2003optimal,eric2019optimal4444}. For $h_4(d)$ we proceed by first noting that $\sum^{v}_{i=1}\textbf{f}(i)\textbf{f}(i)^{\top}=\frac{v-1}{2}\textbf{M}$ and $\sum_{i\neq j}\textbf{f}(i)\textbf{f}(j)^{\top}=-\frac{v-1}{2}\textbf{M}$. \par

For the third-order interactions we consider attributes $k$, $\ell$, $m$ and $r$, say, and distinguish between pairs $(i_ki_{\ell}i_mi_r)$ and $(j_kj_{\ell}j_mj_r)$ in which all the four associated attributes $k$, $\ell, m$ and $r$ differ, pairs $(i_ki_{\ell}i_mi_r)$ and $(j_kj_{\ell}j_mj_r)$ which differ in three of these attributes $k$, $\ell$ and $m$, say, pairs $(i_ki_{\ell}i_mi_r)$ and $(j_kj_{\ell}j_mj_r)$ which differ in two of these attributes $k$ and $\ell$, say, and finally, pairs $(i_ki_{\ell}i_mi_r)$ and $(j_kj_{\ell}j_mj_r)$ which differ in only one of the attributes $k$, say. Hence 

\begin{align}\label{eq:4.32}
&\sum_{i_{k}\neq j_{k}}\sum_{i_{\ell}\neq j_{\ell}}\sum_{i_{m}\neq j_{m}}\sum_{i_{r}\neq j_{r}}(\textbf{f}(i_{k})\otimes\textbf{f}(i_{\ell})\otimes\textbf{f}(i_{m})\otimes\textbf{f}(i_{r})-\textbf{f}(j_{k})\otimes\textbf{f}(j_{\ell})\otimes\textbf{f}(j_{m})\otimes\textbf{f}(j_{r})) \nonumber\\
&\qquad\qquad\qquad\cdot(\textbf{f}(i_{k})\otimes\textbf{f}(i_{\ell})\otimes\textbf{f}(i_{m})\otimes\textbf{f}(i_{r})-\textbf{f}(j_{k})\otimes\textbf{f}(j_{\ell})\otimes\textbf{f}(j_{m})\otimes\textbf{f}(j_{r}))^{\top} \nonumber\\
&=2(v-1)^{4}\sum^{v}_{i_{k}=1}\textbf{f}(i_{k})\textbf{f}(i_{k})^{\top}\otimes\sum^{v}_{i_{\ell}=1}\textbf{f}(i_{\ell})\textbf{f}(i_{\ell})^{\top}\otimes\sum^{v}_{i_{m}=1}\textbf{f}(i_{m})\textbf{f}(i_{m})^{\top}\otimes\sum^{v}_{i_{m}=1}\textbf{f}(i_{r})\textbf{f}(i_{r})^{\top}  \nonumber\\
&\qquad~~~-2\sum_{i_{k}\neq j_{k}}\textbf{f}(i_{k})\textbf{f}(j_{k})^{\top}\otimes\sum_{i_{\ell}\neq j_{\ell}}\textbf{f}(i_{\ell})\textbf{f}(j_{\ell})^{\top}\otimes\sum_{i_{m}\neq j_{m}}\textbf{f}(i_{m})\textbf{f}(j_{m})^{\top}\otimes\sum_{i_{r}\neq j_{r}}\textbf{f}(i_{r})\textbf{f}(j_{r})^{\top}   \nonumber\\
&=\frac{1}{8}v(v-1)^4(v^3-4v^2+6v-4)\textbf{M}^{\otimes q},
\end{align}
also
\begin{align}\label{eq:4.33}
&\sum_{i_{k}\neq j_{k}}\sum_{i_{\ell}\neq j_{\ell}}\sum_{i_{m}\neq j_{m}}\sum_{i_{r}= j_{r}}(\textbf{f}(i_{k})\otimes\textbf{f}(i_{\ell})\otimes\textbf{f}(i_{m})\otimes\textbf{f}(i_{r})-\textbf{f}(j_{k})\otimes\textbf{f}(j_{\ell})\otimes\textbf{f}(j_{m})\otimes\textbf{f}(j_{r})) \nonumber\\
&\qquad\qquad\qquad\cdot(\textbf{f}(i_{k})\otimes\textbf{f}(i_{\ell})\otimes\textbf{f}(i_{m})\otimes\textbf{f}(i_{r})-\textbf{f}(j_{k})\otimes\textbf{f}(j_{\ell})\otimes\textbf{f}(j_{m})\otimes\textbf{f}(j_{r}))^{\top} \nonumber\\
&=2(v-1)^{3}\sum^{v}_{i_{k}=1}\textbf{f}(i_{k})\textbf{f}(i_{k})^{\top}\otimes\sum^{v}_{i_{\ell}=1}\textbf{f}(i_{\ell})\textbf{f}(i_{\ell})^{\top}\otimes\sum^{v}_{i_{m}=1}\textbf{f}(i_{m})\textbf{f}(i_{m})^{\top}\otimes\sum^{v}_{i_{m}=1}\textbf{f}(i_{r})\textbf{f}(i_{r})^{\top}  \nonumber\\
&\qquad-2\sum_{i_{k}\neq j_{k}}\textbf{f}(i_{k})\textbf{f}(j_{k})^{\top}\otimes\sum_{i_{\ell}\neq j_{\ell}}\textbf{f}(i_{\ell})\textbf{f}(j_{\ell})^{\top}\otimes\sum_{i_{m}\neq j_{m}}\textbf{f}(i_{m})\textbf{f}(j_{m})^{\top}\otimes\sum_{i_{r}= j_{r}}\textbf{f}(i_{r})\textbf{f}(j_{r})^{\top}   \nonumber\\
&=\frac{1}{8}v(v-1)^4(v^2-3v+3)\textbf{M}^{\otimes q},
\end{align}
%%%%%%
further
\begin{align}\label{eq:4.34}
&\sum_{i_{k}\neq j_{k}}\sum_{i_{\ell}\neq j_{\ell}}\sum_{i_{m}= j_{m}}\sum_{i_{r}= j_{r}}(\textbf{f}(i_{k})\otimes\textbf{f}(i_{\ell})\otimes\textbf{f}(i_{m})\otimes\textbf{f}(i_{r})-\textbf{f}(j_{k})\otimes\textbf{f}(j_{\ell})\otimes\textbf{f}(j_{m})\otimes\textbf{f}(j_{r})) \nonumber\\
&\qquad\qquad\qquad\cdot(\textbf{f}(i_{k})\otimes\textbf{f}(i_{\ell})\otimes\textbf{f}(i_{m})\otimes\textbf{f}(i_{r})-\textbf{f}(j_{k})\otimes\textbf{f}(j_{\ell})\otimes\textbf{f}(j_{m})\otimes\textbf{f}(j_{r}))^{\top} \nonumber\\
&=2(v-1)^{2}\sum^{v}_{i_{k}=1}\textbf{f}(i_{k})\textbf{f}(i_{k})^{\top}\otimes\sum^{v}_{i_{\ell}=1}\textbf{f}(i_{\ell})\textbf{f}(i_{\ell})^{\top}\otimes\sum^{v}_{i_{m}=1}\textbf{f}(i_{m})\textbf{f}(i_{m})^{\top}\otimes\sum^{v}_{i_{m}=1}\textbf{f}(i_{r})\textbf{f}(i_{r})^{\top}  \nonumber\\
&\qquad-2\sum_{i_{k}\neq j_{k}}\textbf{f}(i_{k})\textbf{f}(j_{k})^{\top}\otimes\sum_{i_{\ell}\neq j_{\ell}}\textbf{f}(i_{\ell})\textbf{f}(j_{\ell})^{\top}\otimes\sum_{i_{m}= j_{m}}\textbf{f}(i_{m})\textbf{f}(j_{m})^{\top}\otimes\sum_{i_{r}= j_{r}}\textbf{f}(i_{r})\textbf{f}(j_{r})^{\top}   \nonumber\\
&=\frac{1}{8}v(v-1)^4(v-2)\textbf{M}^{\otimes q},
\end{align}
%%%%%%
and, finally
\begin{align}\label{eq:4.35}
&\sum_{i_{k}\neq j_{k}}\sum_{i_{\ell}= j_{\ell}}\sum_{i_{m}= j_{m}}\sum_{i_{r}= j_{r}}(\textbf{f}(i_{k})\otimes\textbf{f}(i_{\ell})\otimes\textbf{f}(i_{m})\otimes\textbf{f}(i_{r})-\textbf{f}(j_{k})\otimes\textbf{f}(j_{\ell})\otimes\textbf{f}(j_{m})\otimes\textbf{f}(j_{r})) \nonumber\\
&\qquad\qquad\qquad\cdot(\textbf{f}(i_{k})\otimes\textbf{f}(i_{\ell})\otimes\textbf{f}(i_{m})\otimes\textbf{f}(i_{r})-\textbf{f}(j_{k})\otimes\textbf{f}(j_{\ell})\otimes\textbf{f}(j_{m})\otimes\textbf{f}(j_{r}))^{\top} \nonumber\\
&=2(v-1)\sum^{v}_{i_{k}=1}\textbf{f}(i_{k})\textbf{f}(i_{k})^{\top}\otimes\sum^{v}_{i_{\ell}=1}\textbf{f}(i_{\ell})\textbf{f}(i_{\ell})^{\top}\otimes\sum^{v}_{i_{m}=1}\textbf{f}(i_{m})\textbf{f}(i_{m})^{\top}\otimes\sum^{v}_{i_{m}=1}\textbf{f}(i_{r})\textbf{f}(i_{r})^{\top}  \nonumber\\
&\qquad-2\sum_{i_{k}\neq j_{k}}\textbf{f}(i_{k})\textbf{f}(j_{k})^{\top}\otimes\sum_{i_{\ell}= j_{\ell}}\textbf{f}(i_{\ell})\textbf{f}(j_{\ell})^{\top}\otimes\sum_{i_{m}= j_{m}}\textbf{f}(i_{m})\textbf{f}(j_{m})^{\top}\otimes\sum_{i_{r}= j_{r}}\textbf{f}(i_{r})\textbf{f}(j_{r})^{\top}   \nonumber\\
&=\frac{1}{8}v(v-1)^4\textbf{M}^{\otimes q},
\end{align}
%%%%%%
where $\textbf{M}^{\otimes q}$, $q=1,2,3,4$ is the $q$-fold Kronecker product of $\textbf{M}$. \par
Now for the given attributes $k$, $\ell$, $m$ and $r$ the pairs with distinct levels in the four attributes occur $\left(\begin{smallmatrix}K-4 \\ S-4\end{smallmatrix}\right)\left(\begin{smallmatrix}S-4 \\ d-4\end{smallmatrix}\right)v^{S-4}(v-1)^{d-4}$ times in $\mathcal{X}^{(S)}_d$, while those which differ in three attributes occur $4\left(\begin{smallmatrix}K-4 \\ S-4\end{smallmatrix}\right)\left(\begin{smallmatrix}S-4 \\ d-3\end{smallmatrix}\right)v^{S-4}(v-1)^{d-3}$ times in $\mathcal{X}^{(S)}_d$, while those which differ in two attributes occur~$6\left(\begin{smallmatrix}K-4 \\ S-4\end{smallmatrix}\right)\left(\begin{smallmatrix}S-4 \\ d-2\end{smallmatrix}\right)$ $v^{S-4}(v-1)^{d-2}$ times in $\mathcal{X}^{(S)}_d$. Finally, those which differ only in one attribute occur $4\left(\begin{smallmatrix}K-4 \\ S-4\end{smallmatrix}\right)\left(\begin{smallmatrix}S-4 \\ d-1\end{smallmatrix}\right)v^{S-3}(v-1)^{d-1}$
times. Hence, the diagonal blocks for the interactions are given by
\begin{equation*}\label{eq:4.35}
\begin{split}
&\frac{1}{N_d}\Big(\begin{matrix}K-4 \\ S-4\end{matrix}\Big)\bigg(\frac{1}{8}\Big(\begin{matrix}S-4 \\ d-4\end{matrix}\Big)v^{S-3}(v-1)^d(v^3-4v^2+6v-4)\textbf{M}^{\otimes q}  \\
&\qquad+\frac{1}{2}\Big(\begin{matrix}S-4 \\ d-3\end{matrix}\Big)v^{S-3}(v-1)^{d+1}(v^2-3v+3)\textbf{M}^{\otimes q} \\
&\qquad+\frac{3}{4}\Big(\begin{matrix}S-4 \\ d-2\end{matrix}\Big)v^{S-3}(v-1)^{d+2}(v-2) \textbf{M}^{\otimes q} \\
&\qquad+\frac{1}{2}\Big(\begin{matrix}S-4 \\ d-1\end{matrix}\Big)v^{S-3}(v-1)^{d+3}\textbf{M}^{\otimes q} 
\bigg)\\
&=\frac{d}{8v^3K(K-1)(K-2)(K-3)}((d-1)(d-2)(d-3)(v^3-4v^2+6v-4)  \\
&\qquad\qquad+4(S-d)(d-1)(d-2)(v^2-3v+3)(v-1)  \\
&\qquad\qquad+6(S-d)(S-d-1)(d-1)(v-1)^2(v-2)  \\
&\qquad\qquad+4(S-d)(S-d-1)(S-d-2)(v-1)^3)\textbf{M}^{\otimes q}. 
\end{split}
\end{equation*}
The off-diagonal elements all vanish because the terms in the corresponding entries sum up to zero due to the effects-type coding.
\end{proof}

%%%%%%%%%
%%%%%%%%%
\begin{proof}[Proof of Theorem~\ref{thrm4}]
First we note that the inverse of the information matrix $\mathbf{M}(\bar{\xi})$ of the design $\bar{\xi}$ is given by

\begin{equation*}
\begin{split}
\mathbf{M}(\bar{\xi})^{-1}= diag( \frac{1}{h_{q}(\bar{\xi})}\mathbf{Id}_{p_q}\otimes \mathbf{M}^{\otimes q})_{q=1,\dots,4}.
\end{split}
\end{equation*}

Now, in view of  \citet[][Theorem 2]{eric2019optimal4444}, it follows that for the regression function associated with the interaction of the attributes $k$, $\ell$, $m$ and $r$, say, we obtain

\begin{equation*}\label{eqn4provingthevariance function}
\begin{split}
&\resizebox{1.0\hsize}{!}{$(\textbf{f}(i_{k})\otimes\textbf{f}(i_{\ell})\otimes\textbf{f}(i_{m})\otimes\textbf{f}(i_{r})-\textbf{f}(j_{k})\otimes\textbf{f}(j_{\ell})\otimes\textbf{f}(j_{m})\otimes\textbf{f}(j_{r}))^{\top}\textbf{M}^{-1}\otimes\textbf{M}^{-1}\otimes\textbf{M}^{-1}\otimes\textbf{M}^{-1}$}  \\
&\qquad\cdot\textbf{f}(i_{k})\otimes\textbf{f}(i_{\ell})\otimes\textbf{f}(i_{m})\otimes\textbf{f}(i_{r})-\textbf{f}(j_{k})\otimes\textbf{f}(j_{\ell})\otimes\textbf{f}(j_{m})\otimes\textbf{f}(j_{r}))\\
&=\textbf{f}(i_{k})^{\top}\textbf{M}^{-1}\textbf{f}(i_{k}) \cdot \textbf{f}(i_{\ell})^{\top}\textbf{M}^{-1}\textbf{f}(i_{\ell})\cdot \textbf{f}(i_{m})^{\top}\textbf{M}^{-1}\textbf{f}(i_{m})\cdot \textbf{f}(i_{r})^{\top}\textbf{M}^{-1}\textbf{f}(i_{r})\\
&\qquad+ \textbf{f}(j_{k})^{\top}\textbf{M}^{-1}\textbf{f}(j_{k})\cdot \textbf{f}(j_{\ell})^{\top}\textbf{M}^{-1}\textbf{f}(j_{\ell})\cdot \textbf{f}(j_{m})^{\top}\textbf{M}^{-1}\textbf{f}(j_{m})\cdot \textbf{f}(j_{r})^{\top}\textbf{M}^{-1}\textbf{f}(j_{r}) \\
&\qquad-\textbf{f}(i_{k})^{\top}\textbf{M}^{-1}\textbf{f}(j_{k})\cdot \textbf{f}(i_{\ell})^{\top}\textbf{M}^{-1}\textbf{f}(j_{\ell})\cdot \textbf{f}(i_{m})^{\top}\textbf{M}^{-1}\textbf{f}(j_{m})\cdot \textbf{f}(i_{r})^{\top}\textbf{M}^{-1}\textbf{f}(j_{r})\\
&\qquad- \textbf{f}(j_{k})^{\top}\textbf{M}^{-1}\textbf{f}(i_{k}) \cdot \textbf{f}(j_{\ell})^{\top}\textbf{M}^{-1}\textbf{f}(i_{\ell})\cdot \textbf{f}(j_{m})^{\top}\textbf{M}^{-1}\textbf{f}(i_{m})\cdot \textbf{f}(j_{r})^{\top}\textbf{M}^{-1}\textbf{f}(i_{r})\\
&=\begin{dcases}\frac{1}{8v^3}(v-1)^4(v^3-4v^2+6v-4) & \text{for \quad$i_{k}\neq j_{k},i_{\ell}\neq j_{\ell},i_{m}\neq j_{m},i_{r}\neq j_{r}$} \\
\frac{1}{8v^3}(v-1)^5(v^2-3v+3)   & \text{for \quad$i_{k}\neq j_{k},i_{\ell}\neq j_{\ell},i_{m}\neq j_{m},i_{r}= j_{r}$} \\
\frac{1}{8v^3}(v-1)^6(v-2)  & \text{for \quad$i_{k}\neq j_{k},i_{\ell}\neq j_{\ell},i_{m}= j_{m},i_{r}= j_{r}$} \\
\frac{1}{8v^3}(v-1)^7 & \text{for \quad $i_{k}\neq j_{k},i_{\ell}= j_{\ell},i_{m}= j_{m},i_{r}= j_{r}$}.
\end{dcases}
\end{split}
\end{equation*} \par
Now for a pair of alternatives $(\textbf{i},\textbf{j})\in\mathcal{X}^{(S)}_d$ of comparison depth $d$: there are $d(d-1)(d-2)(d-3)$ third-order interaction terms for which $(i_{k}i_{\ell}i_mi_r)$ and $(j_{k}j_{\ell}j_mj_r)$ differ in all four attributes $k$, $\ell$, $m$ and $r$, there are $(1/6)(S-d)d(d-1)(d-2)$ third-order interaction terms for which $(i_{k}i_{\ell}i_mi_r)$ and $(j_{k}j_{\ell}j_mj_r)$ differ in exactly three of the associated four attributes, there are $(1/4)(S-d)(S-d-1)d(d-1)$ third-order interaction terms for which $(i_{k}i_{\ell}i_mi_r)$ and $(j_{k}j_{\ell}j_mj_r)$ differ in exactly two of the associated four attributes and finally there are $(1/6)(S-d)(S-d-1)(S-d-2)d$ third-order interaction terms for which $(i_{k}i_{\ell}i_mi_r)$ and $(j_{k}j_{\ell}j_mj_r)$ differ in exactly one of the associated four attributes. Hence, we obtain

\begin{equation*}
\begin{split}
V(d,\bar\xi)&=(\textbf{f}(\textbf{i})-\textbf{f}(\textbf{j}))^{\top}\textbf{M}(\bar\xi)^{-1}(\textbf{f}(\textbf{i})-\textbf{f}(\textbf{j}))\\
&=\frac{d(v-1)}{h_{1}(\bar\xi)}+\frac{d(v-1)^2}{4vh_{2}(\bar\xi)}\begin{pmatrix}(d-1)(v-2)+2(S-d)(v-1)\end{pmatrix}      \\
&\qquad+\frac{d(v-1)^3}{24v^2h_{3}(\bar\xi)}\big((d-1)(d-2)(v^2-3v+3)\\
&\qquad\qquad\qquad\qquad+3(S-d)(d-1)(v-1)(v-2)  \\
&\qquad\qquad\qquad\qquad+3(S-d)(S-d-1)(v-1)^2\big) \\
&\qquad+d(d-1)(d-2)(d-3)\frac{(v-1)^4(v^3-4v^2+6v-4)}{8v^3h_{4}(\bar\xi)}  \\
&\qquad+\frac{4(S-d)d(d-1)(d-2)}{24}\frac{(v-1)^5(v^2-3v+3)}{8v^3h_{4}(\bar\xi)}  \\
&\qquad+\frac{6(S-d)(S-d-1)d(d-1)}{24}\frac{(v-1)^6(v-2)}{8v^3h_{4}(\bar\xi)}  \\
&\qquad+\frac{4(S-d)(S-d-1)(S-d-2)d}{24}\frac{(v-1)^7}{8v^3h_{4}(\bar\xi)} \\
&=\frac{d(v-1)}{h_{1}(\bar\xi)}+\frac{d(v-1)^2}{4vh_{2}(\bar\xi)}\begin{pmatrix}(d-1)(v-2)+2(S-d)(v-1)\end{pmatrix}      \\
&\qquad+\frac{d(v-1)^3}{24v^2h_{3}(\bar\xi)}\big((d-1)(d-2)(v^2-3v+3)\\
&\qquad\qquad\qquad\qquad+3(S-d)(d-1)(v-1)(v-2)  \\
&\qquad\qquad\qquad\qquad+3(S-d)(S-d-1)(v-1)^2\big)\\
&\qquad+\frac{d(v-1)^4}{192v^3h_{4}(\bar\xi)}\Big((d-1)(d-2)(d-3)(v^3-4v^2+6v-4) \\
&\qquad\qquad\qquad\qquad+4(S-d)(d-1)(d-2)(v^2-3v+3)(v-1)  \\
&\qquad\qquad\qquad\qquad+6(S-d)(S-d-1)(d-1)(v-1)^2(v-2)  \\
&\qquad\qquad\qquad\qquad+4(S-d)(S-d-1)(S-d-2)(v-1)^3\Big),
\end{split}
\end{equation*} 
for $(\textbf{i},\textbf{j})\in\mathcal{X}^{(S)}_d$ which proofs the proposed formula. 
\end{proof}
%%%%%%%%%%%%%%%%%%%%%%%%%%%%%%%%%%%%%%%%%%%%

\begin{proof}[Proof of Corollary~\ref{cor_thrm4}]
In view of Theorem~\ref{thrm4} it is sufficient to note that the representation of the variance function follows immediately by inserting the values of $h_{r}(\bar{\xi}_d)$ from Lemma \ref{lemma1} and $p_q={K \choose r}(v-1)^{q}$, $q=1,2,3,4$.
\end{proof}

%\begin{proof}
%In view of Theorem~\ref{thrm19} it is sufficient to note that the representation of the variance function follows immediately by inserting the values of $h_{q}(\xi_d)$ from Lemma \ref{lem9} and $p_q={K \choose q}(v-1)^{q}$, $q=1,2,3$. 
%\end{proof}
% % %55 
\begin{proof}[Proof of Theorem~\ref{theorem5}] 
First we note that the variance function $V(d,\xi^{\ast})$ is a polynomial of degree $4$ in the comparison depth $d$ with negative leading coefficient. Now, by the Kiefer-Wolfowitz equivalence theorem $V(d,\xi^{\ast})\leq p$ for all $d=0,1,\dots,S$. Hence, by the shape of the variance function it follows from \citet[][Theorem 3]{eric2019optimal4444} that $V(d,\xi^{\ast})= p$ may occur only at, at most two adjacent comparison depths $d^*$ and $d^*+1$ or $d_1^*$ and $d_1^*+1$, say, in the interior. 
\end{proof}

  \begin{sidewaystable} 
\begin{table}[H]
\centering
\caption{Values of the variance function $V(d,\xi^\ast)$ for $\xi^{\ast}$ from Table~\ref{tab4.7} for $S=K$ (boldface \textbf{1} corresponds to the optimal comparison depths $d^*$ and $d^*_1$).}\label{tab4} 
\begin{tabular*}{\linewidth}{@{\extracolsep{\fill}}
    *{12}{D{.}{.}{4}}
                }
    \toprule
  & \multicolumn{9}{c}{$d$} \\
    \cmidrule(lr){3-12}
   S & v       &    1     &  2 & 3  &  4  & 5    &    6&  7    &8&9&10                       \\
    \hline
5&2&0.938& \textbf{1}&0.938& \textbf{1}&0.938&&&&&\\ 
&3&0.881&\textbf{1} &0.961&1&0.987&&&&&\\   
&4&0.858&\textbf{1}&0.965&0.985&0.981&&&&&\\   
&5&0.845&\textbf{1}&0.970&0.982&0.980&&&&&\\ 
&6&0.837&\textbf{1}&0.974&0.982&0.981&&&&&\\ 
&7&0.832&\textbf{1}&0.977&0.983&0.982&&&&&\\ 
&8&0.828&\textbf{1}&0.980&0.984&0.983&&&&&\\  \hline
6&2&0.850&\textbf{1}&0.950& 0.950&\textbf{1}&0.850&&&&\\
&3&0.793&\textbf{1}&0.988&0.970&\textbf{1}&0.977&&&&\\ 
&4&0.777&0.999&\textbf{1}&0.977&0.995&0.987&&&&\\
&5&0.570&0.984 &\textbf{1}& 0.979&0.990&0.986&&&&\\  
&6&0.734&0.975&\textbf{1}&0.982&0.989&0.987&&&&\\
&7&0.723&0.969 &\textbf{1}&0.984&0.989&0.988&&&&\\ 
&8&0.715&0.964&\textbf{1}&0.985& 0.989&0.988&&&&\\  \hline
7&2&0.792&\textbf{1}&0.982&0.952&0.982&\textbf{1}&0.792&&&\\
&3&0.723&0.973 &\textbf{1}& 0.972&0.971&0.997&0.965&&&\\   
&4&0.679&0.945 &\textbf{1}&0.984&0.976&0.990&0.980&&&\\   
&5&0.657&0.930&\textbf{1}&0.993&0.983&0.992&0.987&&&\\   
&6&0.643&0.921 &\textbf{1}&0.999&0.989&0.995&0.993&&&\\   
&7&0.634&0.914 &0.998&\textbf{1}&0.991&0.995&0.994&&&\\   
&8&0.625&0.906 &0.994&\textbf{1}&0.991&0.995&0.994&&&\\ \hline
8&2&0.759&0.998&\textbf{1}&0.954&0.954&\textbf{1}&0.998&0.759&&\\ 
&3&0.650&0.928&\textbf{1}&0.990& 0.973&0.981&0.998&0.964&&\\  
&4&0.612&0.898 &0.993&\textbf{1}&0.986&0.984&0.995&0.984&&\\  
&5&0.585&0.873 &0.982&\textbf{1}&0.990&0.986&0.993&0.988&&\\ 
&6&0.567&0.858 & 0.974&\textbf{1}&0.994&0.989&0.994&0.991&&\\   
&7&0.559&0.848&0.969&\textbf{1}&0.997&0.992&0.995&0.994&&\\  
&8&0.552& 0.841 &0.965&\textbf{1}&0.999&0.994&0.996&0.996&&\\ \hline
\end{tabular*}
\raggedleft\footnotesize (To be continued)
    \end{table}
 \end{sidewaystable} 

 \begin{sidewaystable}     
\begin{table}[H]
\centering
\setlength\tabcolsep{0pt}
\caption*{Table \thetable{} (continued)}
\begin{tabular*}{\linewidth}{@{\extracolsep{\fill}}
    *{12}{D{.}{.}{4}}
                }
    \hline
  & \multicolumn{9}{c}{$d$} \\
    \cmidrule(lr){3-12}
   S & v       &    1     &  2 & 3  &  4  & 5    &    6&  7    &8&9& 10                     \\
    \hline
9&2&0.693&0.958&\textbf{1}&0.966&0.945&0.966&\textbf{1}&0.958&0.693&\\
&3&0.596&0.885 &0.989&\textbf{1}& 0.983&0.976&0.987&0.997&0.960&\\   
&4&0.550&0.841& 0.967&\textbf{1}&0.995& 0.986&0.987& 0.995&0.984&\\   
&5&0.528&0.819 &0.954&0.998&\textbf{1}&0.992&0.991& 0.996& 0.992&\\   
&6&0.512&0.801 &0.941&0.993&\textbf{1}&0.994& 0.992&0.996&0.993&\\   
&7&0.501&0.789 &0.932&0.989&\textbf{1}&0.996& 0.993&0.996&0.995&\\ 
&8&0.493& 0.780 &0.927&0.986&\textbf{1}&0.997&0.994&0.996&0.996&\\ \hline
10&2&0.644&0.925& \textbf{1}&0.985&0.958&0.958&0.985&\textbf{1}&0.925&0.644\\
&3&0.544&0.836 &0.965&\textbf{1}&0.994&0.982&0.981&0.991&0.996& 0.960\\   
&4&0.501& 0.791 &0.936&0.991&\textbf{1}&0.993&0.987&0.990&0.996& 0.985\\   
&5&0.478& 0.764 &0.916&0.982&\textbf{1}&0.997& 0.992& 0.993&0.996& 0.992\\   
&6&0.464&0.748 &0.904&0.976&0.999&\textbf{1}&0.996&0.995&0.998&0.995\\  
&7&0.453&0.735 &0.893&0.969& 0.997&\textbf{1}& 0.996&0.995&0.997&0.996\\   
&8&0.446&0.726&0.885&0.965&0.995&\textbf{1}&0.997&0.996&0.996&0.997\\  \bottomrule
\end{tabular*}
    \end{table}   
     \end{sidewaystable}  
\end{appendices}
\end{document}